\documentclass[12pt,reqno]{article}

\usepackage{enumitem}
\usepackage[usenames]{color}
\usepackage{amssymb}
\usepackage{amsmath}
\usepackage{amsthm}
\usepackage{amsfonts}
\usepackage{amscd}
\usepackage{graphicx}
\usepackage{diagbox}
\usepackage{mathrsfs}
\usepackage[title]{appendix}

\usepackage{cprotect}

\usepackage[colorlinks=true,
linkcolor=webgreen,
filecolor=webbrown,
citecolor=webgreen]{hyperref}

\definecolor{webgreen}{rgb}{0,.5,0}
\definecolor{webbrown}{rgb}{.6,0,0}

\usepackage{color}
\usepackage{fullpage}
\usepackage{float}

\usepackage{graphics}
\usepackage{latexsym}
\usepackage{epsf}
\usepackage{breakurl}

\newcommand{\seqnum}[1]{\href{https://oeis.org/#1}{\underline{#1}}}

\def\Enn{\mathbb{N}}

\DeclareMathOperator{\ftmfactoreq}{ftmfactoreq}
\DeclareMathOperator{\isftmrs}{isftmrs}
\DeclareMathOperator{\ftmrsn}{ftmrsn}

\begin{document}

\theoremstyle{plain}
\newtheorem{theorem}{Theorem}
\newtheorem{corollary}[theorem]{Corollary}
\newtheorem{lemma}[theorem]{Lemma}
\newtheorem{proposition}[theorem]{Proposition}

\theoremstyle{definition}
\newtheorem{definition}[theorem]{Definition}
\newtheorem{example}[theorem]{Example}
\newtheorem{conjecture}[theorem]{Conjecture}

\theoremstyle{remark}
\newtheorem{remark}[theorem]{Remark}

\title{Subword complexity of the Fibonacci-Thue-Morse sequence:
the proof of Dekking's conjecture}

\author{Jeffrey Shallit\footnote{Supported by NSERC Grant 2018-04118.}\\
School of Computer Science \\
University of Waterloo \\
Waterloo, ON  N2L 3G1 \\
Canada \\
{\tt shallit@uwaterloo.ca} \\
}

\date{October 20 2020}

\maketitle

\begin{abstract}
Recently F. M. Dekking conjectured the form of the subword complexity function
for the Fibonacci-Thue-Morse sequence.  In this note we prove his
conjecture by purely computational means, using the free software
{\tt Walnut}.
\end{abstract}

\section{Introduction}

Recall that the Fibonacci numbers $F_n$ are defined by
$F_0 = 0$, $F_1 = 1$, and $F_n = F_{n-1} + F_{n-2}$.

In this note we will use the Fibonacci (or ``Zeckendorf'') 
representation of natural numbers \cite{Lekkerkerker:1952,Zeckendorf:1972}.
This is a map from $\Enn$ to the set of
binary strings, written $(n)_F$, such that
\begin{itemize}
\item[(a)] If $(n)_F = w$, and $t = |w|$,
then $n = \sum_{1 \leq i \leq t} w[i] F_{t+2-i}$,
\item[(b)] $w[1]$ (if it exists) is $1$;
\item[(c)] $w$ does not have two consecutive $1$'s.
\end{itemize}
The inverse map is written $[w]_F$.   Thus, Fibonacci representation essentially
writes a natural number as a sum of Fibonacci numbers, with no two consecutive
Fibonacci numbers appearing in the sum.

The so-called {\it Fibonacci-Thue-Morse} sequence
$({\bf ftm}[n])_{n \geq 0}$ is defined to be the sum, taken modulo $2$,
of the bits of the Fibonacci representation of $n$.
See, for example, \cite[Examples 7.8.2, 7.8.4]{Allouche&Shallit:2003}
and \cite{Ferrand:2007}, where this sequence is discussed.  It is sequence
\seqnum{A095076} in the {\it On-Line Encyclopedia of Integer Sequences}
(OEIS).  The first few terms of the sequence are as follows:
\begin{center}
\begin{tabular}{c|ccccccccccccccccccccc}
$n$ & 0& 1& 2& 3& 4& 5& 6& 7& 8& 9&10&11&12&13&14&15&16&17&18&19\\
\hline
${\bf ftm}[n]$ &  0&1&1&1&0&1&0&0&1&0&0&0&1&1&0&0&0&1&0&1\\
\end{tabular}
\end{center}

The sequence $\bf ftm$ is {\it Fibonacci-automatic}.   This means
there is a deterministic finite automaton with output (DFAO)
that, on input the Fibonacci representation of
$n$, computes ${\bf ftm}[n]$.  The state name $q$ and the corresponding output
$a$ is written as $q/a$, and the output of the machine is the output of the last
state reached after following the transitions corresponding to the input,
starting from the initial state $q_0$.   In fact, it is easy to see that
this DFAO has $4$ states, as follows:
\begin{figure}[H]
\vskip -.5in
\begin{center}
\includegraphics[width=5in]{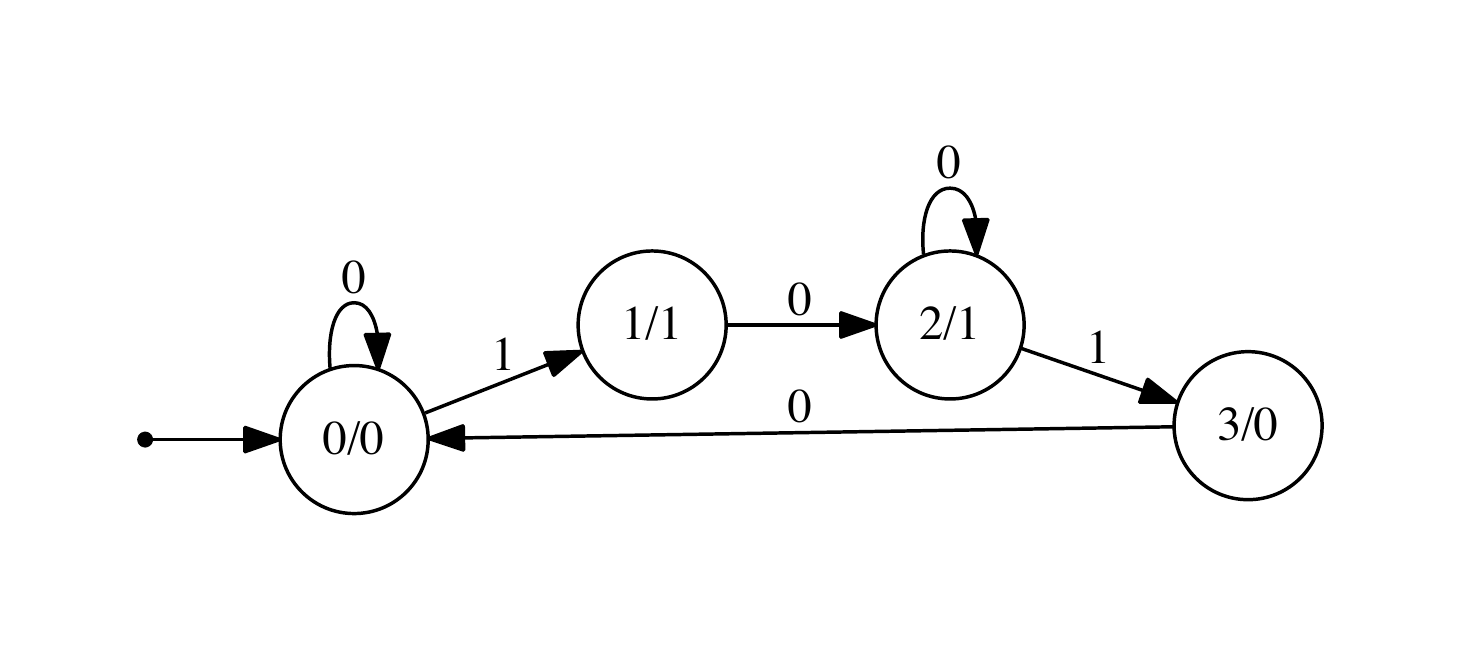}
\end{center}
\vskip -.5in
\caption{Fibonacci automaton for the sequence $\bf ftm$}
\label{fig0}
\end{figure}
\noindent Each state is labeled with the syntax ``name/output''.

The {\it subword complexity} function $\rho_{\bf s}(n)$ of a sequence $\bf s$ is
defined to be the number of distinct length-$n$ contiguous blocks appearing
in $\bf s$.   Recently Dekking \cite{Dekking:2020} gave a conjecture on
the form of $\rho_{\bf ftm} (n)$.  In this note we prove this conjecture.

\section{Dekking's conjecture}

One form of Dekking's conjecture involves the first difference 
$d(n) := \rho_{\bf ftm}(n+1) - \rho_{\bf ftm} (n)$ of the subword complexity
function.  It is well-known, and easy to see, that $d(n)$ counts the number of
{\it right-special} factors, where a factor $x$ is right-special if
both $x0$ and $x1$ are factors of $\bf ftm$.

\begin{conjecture}[Dekking]
\begin{align*}
(d(n))_{n \geq 0} &= 1,2,4,6,10,6,6,8,6,8,8,6,6,6,6,8,8,8,6,6,6, \ldots \\
&= 1,2,4,6,10,6^2,8^1,6^1,8^2,6^4,8^3,6^4,8^5,6^9,8^8,6^{12},8^{13} \ldots \\
&= (1,2,4,6,10) \prod_{i \geq 2} 6^{F_i + (-1)^i}\, 8^{F_i}.
\end{align*}
\end{conjecture}
Here the exponents on the numbers denote repetition factors, and the
$\prod$ symbol denotes concatenation.

\section{Outline of the proof}

We can prove Dekking's conjecture via purely computational means, using
a decision procedure implemented in the free
software {\tt Walnut} \cite{Mousavi:2016}.

There are five steps in the proof.

\begin{enumerate}
\item We create a first-order logical formula specifying that
the length-$n$ factor ${\bf ftm}[i..i+n-1]$ beginning at position $i$
is a right-special factor.

\item Using the ideas in \cite{Mousavi&Schaeffer&Shallit:2016}, we can create a deterministic finite
automaton (DFA) accepting the Fibonacci representation of the
pairs $(i,n)$ such that
\begin{itemize}
\item[(a)] ${\bf ftm}[i..i+n-1]$ is a right-special
factor; and
\item[(b)] ${\bf ftm}[i..i+n-1]$ is the earliest occurrence
of that factor in $\bf ftm$.
\end{itemize}

\item Using the ideas in \cite{Du&Mousavi&Schaeffer&Shallit:2016},
we can read off, from the automaton in the previous step, a
{\it linear representation\/} for $d(n)$.   This consists
of vectors $v, w$ and a matrix-valued morphism $\gamma$ such
that $d(n) = v \gamma( (n)_F ) w$ for all $n \geq 0$.
The {\it rank\/} of the linear representation is the size of
$v$.

\item Using the so-called ``semigroup trick'' 
\cite{Du&Mousavi&Schaeffer&Shallit:2016}, we can prove
that the semigroup generated by the two matrices
$\gamma(0), \gamma(1)$ is finite.   Furthermore, we can
create a Fibonacci-DFAO computing $d(n)$.  

\item This DFAO can be minimized using
the standard methods \cite{Hopcroft&Ullman:1979},
and Dekking's result can
be read off the result.
\end{enumerate}

\section{Details of the proof}

Now we provide a few more details about each step.

\begin{enumerate}

\item The logical formula is created in several steps.  First, we need
a formula asserting that the factor ${\bf ftm}[i..i+n-1]$ is equal to the
factor ${\bf ftm}[j..j+n-1]$.
One obvious way to write this would be
\begin{equation}
\forall t \ (t<n) \implies {\bf ftm}[i+t]= {\bf ftm}[j+t].
\label{first}
\end{equation}
However, for computational purposes, it turns out to be more efficient to formulate this assertion somewhat differently:
$$ \ftmfactoreq(i,n) :=  \forall t,u\ (t \geq i \ \wedge\ t<i+n \ \wedge\ i+u=t+j) \implies {\bf ftm}[t]= {\bf ftm} [u] .$$
It is easy to see this is logically equivalent to \eqref{first}.

Next, we create a formula asserting that ${\bf ftm}[i..i+n-1]$ is
right-special:
$$
\isftmrs(i,n) := \exists j,k \ \ftmfactoreq(i,j,n) \ \wedge\ \ftmfactoreq(i,k,n) 
\ \wedge\ {\bf ftm}[j+n] \not= {\bf ftm}[k+n] .$$
Finally, we create a formula asserting that ${\bf ftm}[i..i+n-1]$
is right-special and also is the first occurrence of that factor:
$$ \ftmrsn(i,n) := \isftmrs(i,n) \ \wedge\  \forall j\ (j<i) 
\implies \neg \ftmfactoreq(i,j,n) .$$

\item We can translate 
the first-order formulas above into automata using the {\tt Walnut} system.
The created automata recognize the Fibonacci representations of the values of
the free variables in the formulas that make the formula true.
We do this with the following {\tt Walnut} commands, which are more
or less verbatim translations of the predicates above.
\begin{verbatim}
def ftmfactoreq "?msd_fib At,u (t>=i & t<i+n & i+u=t+j) => FTM[t]=FTM[u]":
def isftmrs "?msd_fib Ej,k $ftmfactoreq(i,j,n) & $ftmfactoreq(i,k,n) 
     & FTM[j+n] != FTM[k+n]":
def ftmrsn "?msd_fib $isftmrs(i,n) &  Aj (j<i) => ~$ftmfactoreq(i,j,n)":
\end{verbatim}
The automata that are created by {\tt Walnut} are DFA's (deterministic finite automata).
The DFA for {\tt ftmfactoreq} has 91 states, the
DFA for {\tt isftmrs} has 45 states, and the
DFA for {\tt ftmrsn} has 46 states.

\item We can create the linear representation with the {\tt Walnut}
command
\begin{verbatim}
eval ftmatrix n "?msd_fib $ftmrsn(i,n)":
\end{verbatim}
This gives us a linear representation $(v, \gamma, w)$ of rank $46$, which is
stored in the file {\tt ftmatrix.mpl}.

For technical reasons,
this representation now needs to be adjusted by replacing
$v$ with $v \gamma(000)$.  

\item The ``semigroup trick'' gives a DFAO with 75 states.

\item  After minimization, the Fibonacci-DFAO produced is below.
\begin{figure}[H]
\begin{center}
\includegraphics[width=6.5in]{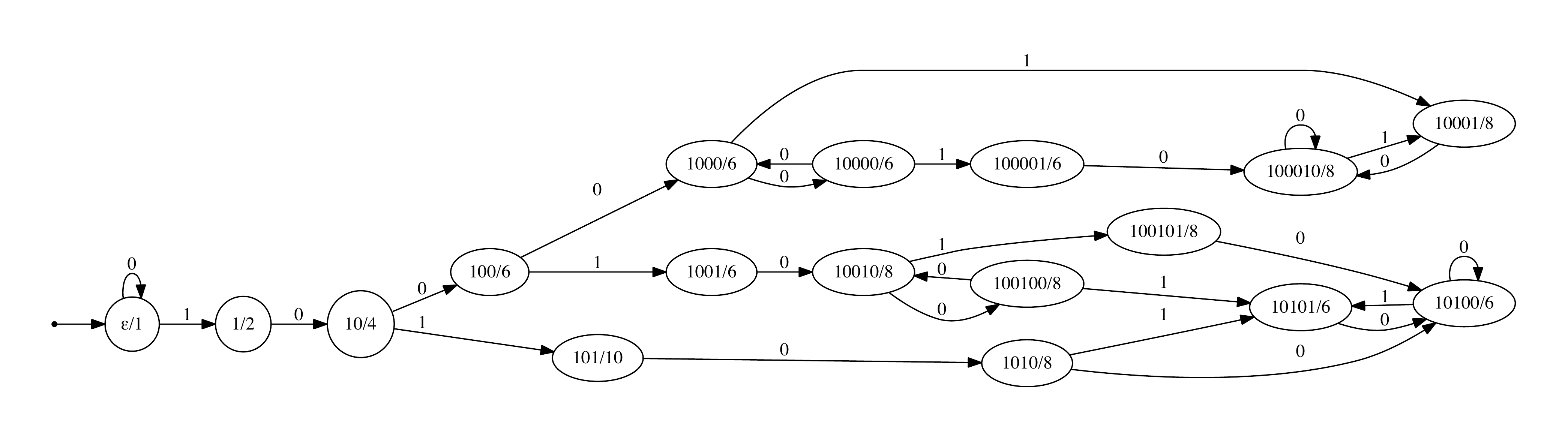}
\end{center}
\caption{Fibonacci-DFAO computing $d(n)$}
\label{fig1}
\end{figure}

\end{enumerate}

\section{Deducing Dekking's result from the automaton}

\begin{theorem}
We claim that 
\begin{align*}
d(n) &= 8 \text{ for } F_{2m+1} + 2 \leq n \leq F_{2m+1} + F_{2m-2} + 1 \\
d(n) &= 6 \text{ for } F_{2m+1} + F_{2m-2} + 2 \leq n \leq F_{2m+2} \\
d(n) &= 8 \text{ for } F_{2m+2} + 1 \leq n \leq F_{2m+2} + F_{2m-1} \\
d(n) &= 6 \text{ for } F_{2m+2} + F_{2m-1} + 1 \leq n \leq F_{2m+3} + 1
\end{align*}
for $m \geq 2$.
\end{theorem}

\begin{proof}
We can prove this with {\tt Walnut}, too.
The first step is to take the DFAO given in Figure~\ref{fig1},
and add it to Walnut's {\tt Word Automata Library}.
We'll use the name {\tt FTMD} for this DFAO.  It is given in Appendix A.

Next, we will create individual formulas for each of the intervals
given in the statement of the theorem; we'll call them
{\tt interval1}, $\ldots$, {\tt interval4}.   In order to specify them
we need ways to specify that $x$ is a Fibonacci number, an odd-indexed
Fibonacci number, an even-indexed Fibonacci number, or that two numbers
represent consecutive Fibonacci numbers.   These are done,
respectively, with the formulas {\tt fibo}, {\tt oddfib}, {\tt evenfib},
and {\tt consecfib}.    The first three of these use the standard formalism of
regular expressions \cite[\S 2.5]{Hopcroft&Ullman:1979}; for example,
{\tt evenfib} specifies that the Fibonacci representation must consist of
$1$ followed by an even number of $0$'s, which corresponds to $F_{2n}$ for
some $n$.  The last is more complicated; basically, it asserts that $t$ and $u$
are both Fibonacci numbers, and there is no Fibonacci number strictly between
$t$ and $u$.
We also need a formula specifying that
$x$ is of the form $F_n + F_{n-3}$; this is done with a regular expression
in {\tt fib1001}.   This gives the following formulas.
\begin{verbatim}
reg fibo msd_fib "0*10*":
reg oddfib msd_fib "0*10(00)*":
reg evenfib msd_fib "0*1(00)*":
def consecfib "?msd_fib (t<u) & $fibo(t) & $fibo(u) & Av (t<v & v<u) => ~$fibo(v)":
reg fib1001 msd_fib "0*10010*":
\end{verbatim}

From these we can create formulas for our four kinds of intervals:
\begin{verbatim}
def interval1 "?msd_fib Et,u,y $oddfib(t) & $evenfib(u) & $consecfib(t,u) &
    $fib1001(y) & (t<y) & (y<=u) & (t+2<=n) & (n<=y+1)":
def interval2 "?msd_fib Et,u,y $oddfib(t) & $evenfib(u) & $consecfib(t,u) &
    $fib1001(y) & (t<y) & (y<=u) & (y+2<=n) & (n<=u)":
def interval3 "?msd_fib Ew,x,z $evenfib(x) & $oddfib(z) & $consecfib(x,z) &
    $fib1001(w) & (x<w) & (w<=z) & (x+1<=n) & (n<=w)":
def interval4 "?msd_fib Ew,x,z $evenfib(x) & $oddfib(z) & $consecfib(x,z) &
    $fib1001(w) & (x<w) & (w<=z) & (w+1<=n) & (n<=z+1)":
\end{verbatim}

As a double-check, we can now verify with {\tt Walnut} that our four defined intervals cover all integers $n \geq 7$ with
the formula {\tt all},
\begin{verbatim}
def all "?msd_fib $interval1(n) | $interval2(n) | $interval3(n) | $interval4(n)":
\end{verbatim}
and the resulting automaton accepts all $n \geq 7$.

Finally, we perform four tests, one for each interval, checking each of
the assertions in the statement of the theorem.  Each of these tests
returns {\tt true}, so the theorem is proven.
\begin{verbatim}
eval test1 "?msd_fib An $interval1(n) => FTMD[n]=@8":
eval test2 "?msd_fib An $interval2(n) => FTMD[n]=@6":
eval test3 "?msd_fib An $interval3(n) => FTMD[n]=@8":
eval test4 "?msd_fib An $interval4(n) => FTMD[n]=@6":
\end{verbatim}
\end{proof}

\begin{corollary}
Dekking's conjecture is correct.
\end{corollary}

\section{Going further}

It should be possible to prove similar results for the analogues of the Fibonacci-Thue-Morse sequence for other kinds of numeration systems, such as those built from the
recurrence $E_n = k E_{n-1} + E_{n-2}$ for an integer $k \geq 3$.

\begin{appendices}
\section{{\tt Walnut} file {\tt FTMD.txt}}

{\scriptsize
\begin{verbatim}
msd_fib
0 1
0 -> 0
1 -> 1
1 2
0 -> 2
2 4
0 -> 3
1 -> 4
3 6
0 -> 5
1 -> 6
4 10
0 -> 7
5 6
0 -> 8
1 -> 9
6 6
0 -> 10
7 8
0 -> 11
1 -> 12
8 6
0 -> 5
1 -> 13
9 8
0 -> 14
10 8
0 -> 15
1 -> 16
11 6
0 -> 11
1 -> 12
12 6
0 -> 11
13 6
0 -> 14
14 8
0 -> 14
1 -> 9
15 8
0 -> 10
1 -> 12
16 8
0 -> 11
\end{verbatim}
}
\end{appendices}

\end{document}